\documentclass[11pt]{article}
\usepackage[margin=1in]{geometry}
\usepackage[T1]{fontenc}
\usepackage{lmodern,amsmath,amssymb,amsthm,mathtools,enumitem,tabularray,needspace}
\usepackage[dvipsnames]{xcolor}
\definecolor{xrcomment}{HTML}{D64A8C}

\usepackage[colorlinks=true,linkcolor=magenta!70!black,citecolor=blue!60!black,urlcolor=blue!60!black]{hyperref}
\usepackage[nameinlink,noabbrev]{cleveref}
\makeatletter
\newcommand{\citegroup}[1]{%
  \begingroup
  \def\@cite##1##2{##1\if@tempswa, ##2\fi}%
  [#1]%
  \endgroup
}
\makeatother
\hypersetup{pdftitle={Almost Optimal FPT Inapproximability for k-SetCover},
 pdfauthor={Venkatesan Guruswami and Xuandi Ren},
 pdfsubject={Parameterized inapproximability of Set Cover}}
\setlist[enumerate]{itemsep=3pt,topsep=4pt,leftmargin=2em,label=(\roman*)}
\newtheorem{theorem}{Theorem}[section]
\newtheorem{lemma}[theorem]{Lemma}

\theoremstyle{definition}
\newtheorem{openq}{Open Question}
\newtheorem{definition}[theorem]{Definition}
\theoremstyle{remark}

\newcommand{\OPT}{\operatorname{OPT}}
\newcommand{\fpt}{\textnormal{\textsf{FPT}}}
\newcommand{\wone}{\textnormal{\textsf{W[1]}}}
\newcommand{\wtwo}{\textnormal{\textsf{W[2]}}}
\newcommand{\hypoeth}{\textnormal{\textsf{ETH}}}
\newcommand{\hyposeth}{\textnormal{\textsf{SETH}}}
\newcommand{\hypogapeth}{\textnormal{\textsf{Gap-ETH}}}
\DeclareRobustCommand{\Clique}{\textnormal{\textsc{Clique}}}
\DeclareRobustCommand{\SetCover}{\textnormal{\textsc{SetCover}}}
\DeclareRobustCommand{\MMSA}{\textnormal{\textsc{MMSA}}}
\DeclareRobustCommand{\CSP}{\textnormal{\textsc{CSP}}}
\DeclareRobustCommand{\SAT}{\textnormal{\textsc{SAT}}}
\title{Almost Optimal FPT Inapproximability for $k$-\SetCover{}\thanks{Work supported in part by NSF grant CCF-2211972 and DOD Advanced
Research Projects Agency grant HR0011262E031.}\\[0.8em]
{\normalsize Reported by}}
\author{Venkatesan Guruswami\\[3pt]
{\small\texttt{venkatg@berkeley.edu}}\\[3pt]
{\small Departments of EECS \& Mathematics}\\
{\small UC Berkeley}
\and Xuandi Ren\\[3pt]
{\small\texttt{xuandi\_ren@berkeley.edu}}\\[3pt]
{\small Department of EECS}\\
{\small UC Berkeley}}
\date{}
\usepackage{todonotes}
\begin{document}
\hypersetup{pageanchor=false}
\begin{titlepage}
\maketitle
\thispagestyle{empty}
\vspace{8mm}
\begin{abstract}
We show that $\bigl(\frac{\log n}{\log\log n}\bigr)$-approximate parameterized $k$-\SetCover{} is $\wone$-hard, and has no $n^{o(k/\log k)}$-time algorithms under \hypoeth{}. This improves upon the previous best factors $\bigl(\frac{\log n}{\log\log n}\bigr)^{1/k}$ in (Lin, 2019)
and $(\log n)^{1/\operatorname{poly}(k)}$ in (Karthik, Laekhanukit, and Manurangsi, 2019). Here $k$ is the yes-case guarantee and $n$ is the number of candidate sets.
While the best approximation ratio is still $O(\log n)$ via the greedy algorithm, closing this $1/k$ gap in the exponent has been a longstanding open
problem; we remove this loss via a simple direct
reduction.

\smallskip 
The construction is self-contained and does not rely on the parameterized inapproximability hypothesis (PIH). Starting with sparse parameterized
2-\CSP{} instances (Karthik, Marx, Pilipczuk, and Souza, 2024), we build a
monotone CNF formula, which is equivalent to a \SetCover{} instance.
To obtain a $k$-versus-$h$ gap, the reduction enumerates all hash functions
from $\Sigma$ to $[2h]$ and all unsatisfiable 2-\CSP{} instances on the same
constraint graph with alphabet $[2h]$.
For each such instance, it asks for a certificate that the hashed label pairs are not all contained in that instance. Perfect hashing makes this enumeration efficient
for $h=\log n/\log\log n$.
\end{abstract}

\vspace{3mm}
\begin{quotation}\small
\noindent\textbf{AI Disclosure.}
The proof was discovered by GPT-6 Astra in response to the authors' prompts about the parameterized Minimum Monotone Satisfying Assignment (MMSA) problem. The authors checked the proofs and wrote the paper; they take responsibility for the correctness of the results but claim no intellectual credit for them. 
\end{quotation}
\end{titlepage}
\hypersetup{pageanchor=true}
\setcounter{page}{1}
\section{Introduction}\label{sec:intro}

\SetCover{} is a fundamental problem in combinatorial optimization. It is one of Karp's first 21 NP-complete problems~\cite{karp72}.
Given a collection of sets over a universe $U$, the goal is to find a
smallest subcollection whose union is $U$. 
The simple formulation and the tight connection to other combinatorial problems have
made \SetCover{} a central example for studying hardness of approximation. The polynomial-time approximability of \SetCover{} is by now well
understood. The greedy algorithm achieves a
$(1+\ln |U|)$-approximation~\cite{johnson,chvatal}. On the hardness
side, Feige~\cite{feige98} ruled out
$(1-\varepsilon)\ln |U|$-approximation in polynomial time for every fixed
$\varepsilon\in(0,1)$, assuming $\mathsf{NP}$ is not in quasi-polynomial time.
Later works established NP-hardness for the same approximation ratio~\cite{mos15,ds14}.
Thus the logarithmic approximation guarantee
is essentially optimal for polynomial-time algorithms, unless
$\mathsf{P}=\mathsf{NP}$.

Parameterized complexity offers a complementary approach to $\mathsf{NP}$-hard problems by identifying a parameter $k$ and seeking algorithms whose running time depends arbitrarily on $k$ but only polynomially on the input size $n$.  An algorithm is said to be fixed-parameter tractable
(FPT) if it runs in time $f(k)n^{O(1)}$, where $f$ can be any computable function. Parameterized by the desired number of sets, $k$-\SetCover{}
is a canonical $\wtwo$-complete problem~\cite{df99}, so it has no FPT
algorithm unless $\wtwo=\fpt$. 
It is then natural to ask how good an approximation algorithm can be when FPT running time is allowed.

A long line of work has pursued this direction. Chen and Lin~\cite{cl16} ruled out any constant-factor FPT approximation for $k$-\SetCover{} under $\wone\neq\fpt$, and a
$(\log k)^{1/4-\varepsilon}$-factor FPT approximation for every fixed $\varepsilon\in(0,1/4)$ under \hypoeth{}.
Chalermsook, Cygan, Kortsarz, Laekhanukit, Manurangsi, Nanongkai, and
Trevisan~\cite{cck17} then established \textit{total FPT inapproximability} for $k$-\SetCover{} by ruling out any $g(k)$-factor approximation even in $f(k)n^{o(k)}$ time under \hypogapeth{}. Karthik, Laekhanukit, and Manurangsi~\cite{klm19} and Lin~\cite{lin19} used two different approaches to weaken the \hypogapeth{} assumption to gap-free hypotheses, giving $(\log n)^{o(1)}$-factor\footnote{The factor stated in~\cite[Theorem~1.3]{klm19} is $(\log n)^{1/\operatorname{poly}(k)}$. When only FPT running time is required, padding strengthens this to $(\log n)^{\alpha(k)}$ for every computable positive function $\alpha(k)=o(1)$; see~\cite[Lemma~2.8]{glrz}.} FPT inapproximability under $\wone \ne \fpt$, and $(\log n)^{1/\operatorname{poly}(k)}$ and $\left(\log n/\log\log n\right)^{1/k}$ factors, respectively, in $f(k)n^{o(k)}$ time under \hypoeth{}. Lin, Ren, Sun, and Wang~\cite{lrsw23} later showed the $\wtwo$-hardness of approximating $k$-\SetCover{} within any constant factor. While these hardness results seem satisfactory, we remark that the best FPT approximation algorithm is still the classical greedy one, which is an $O(\log n)$-approximation. There is still a gap in the exponent between the algorithm and the hardness results: $1$ versus $o(1)$ under $\wone \ne\fpt$ and $1$ versus $1/k$ under \hypoeth{}. Closing this gap has been listed as an important open problem in the literature~\citegroup{\cite[Section~4]{lin19};
\cite[Open Question~4]{fklm20};
\cite[Question~2]{lrsw23}}.

\begin{openq}[{\cite[Open Question~4]{fklm20}}]
Is there a $(\log n)^{1-o(1)}$ factor approximation algorithm for
$k$-\SetCover{} running in time $n^{k-0.1}$?
\end{openq}

\paragraph{Our results.}
In this paper, we remove the aforementioned gap in the exponent and establish FPT inapproximability for $k$-\SetCover{} with a factor of $\frac{\log n}{\log\log n}$. Our reduction gives
\wone{}-hardness and, under \hypoeth{}, excludes
$f(k)n^{o(k/\log k)}$-time algorithms for every computable function
$f$. See \Cref{thm:main,thm:main-eth} below and \Cref{tab:comparison} for a comparison between our results and the previous state of the art.

\Needspace{8\baselineskip}
\begin{theorem}\label{thm:main}
Assuming $\wone\neq\fpt$, there is no FPT algorithm that approximates
$k$-\SetCover{} within a $\frac{\log n}{\log\log n}$ factor.
\end{theorem}

\begin{theorem}\label{thm:main-eth}
Assuming \hypoeth{}, for any computable function $f$, there is no
$f(k)n^{o(k/\log k)}$-time algorithm that approximates
$k$-\SetCover{} within a $\frac{\log n}{\log\log n}$ factor.
\end{theorem}

We remark that the near-logarithmic FPT inapproximability for $k$-\SetCover{} was obtained
independently and almost simultaneously by Lin and Zheng~\cite{LZ26}, and Karthik~\cite{K26}, using potentially different reductions.
They also obtain tight running-time lower bounds under \hypoeth{}, \hyposeth{},
and the $k$-\textsc{SUM} hypothesis.

\begin{table}[htbp]
\centering
\small
\begin{tblr}{
 colspec={cccc},
 cells={c,m},
 rows={ht=2.6em},
 row{1}={font=\bfseries,ht=2.3em},
 colsep=5pt,rowsep=3pt,
 column{1}={leftsep=0pt},
 column{4}={rightsep=0pt},
 vline{2-4}={0.4pt},
 hline{2-Z}={0.4pt}
}
Assumption & Inapproximability Factor & Runtime Lower Bound & Reference\\
\SetCell[r=3]{c,m} $\wone\neq\fpt$
 & \SetCell[r=2]{c,m} $(\log n)^{o(1)}$
 & \SetCell[r=3]{c,m} $f(k)n^{O(1)}$
 & \cite[Theorem~4]{lin19}\\
 & & & \cite[Theorem~1.3]{klm19}\\
 & $\dfrac{\log n}{\log\log n}$ & & \Cref{thm:main}\\
\SetCell[r=3]{c,m} \hypoeth{}
 & $\left(\dfrac{\log n}{\log\log n}\right)^{1/k}$
 & \SetCell[r=2]{c,m} $f(k)n^{o(k)}$
 & \cite[Theorem~2]{lin19}\\
 & $(\log n)^{1/\operatorname{poly}(k)}$ & & \cite[Theorem~1.4]{klm19}\\
 & $\dfrac{\log n}{\log\log n}$ & $f(k)n^{o(k/\log k)}$
 & \Cref{thm:main-eth}\\
\end{tblr}
\caption{A summary of FPT inapproximability results for $k$-\SetCover{}.}
\label{tab:comparison}
\end{table}

\paragraph{Connection to minimum monotone satisfying assignment.}
\SetCover{} is equivalent to the depth-two Minimum Monotone Satisfying Assignment (\MMSA{}) problem~\cite{abmp01,ds04}, in which we are given a monotone CNF formula: each Boolean variable corresponds to a candidate set and each clause corresponds to an element of the universe. The goal is to minimize the number of Boolean variables set to $1$ while satisfying the formula.
This correspondence preserves the optimum. Parameterized
inapproximability of $k$-\MMSA{} and connections between variants of different depths were studied
in~\cite{mar13,glrz}. In the following, we use the above CNF view to describe our constructions for ease of presentation.

\paragraph{Proof overview.}
As a proof of concept, we present a direct reduction from $k$-\Clique{}, which suffices for the \wone{}-hardness. Let $G=(V,E)$ be a $k$-\Clique{} instance. Put
$m=\binom{k}{2}$, and fix an integer $h\ge m$.
Introduce one Boolean variable $x_{uv}$ for each edge $uv\in E$,
so the total number of Boolean variables is $n=|E|$.
An edge set $X\subseteq E$ specifies an assignment by setting
$x_{uv}=1$ whenever $uv\in X$.
Take a perfect hash family $\Phi$ from $V$ to $[2h]$:
for every $A\subseteq V$ of size at most $2h$, some
$\phi\in\Phi$ is injective on $A$.
Let $\mathcal H$ be the collection of all \emph{bad graphs}: simple graphs
on $[2h]$ with at most $h$ edges and no $k$-clique.

For each $\phi\in\Phi$ and
$H\in\mathcal H$, form the clause
\[
 \begin{gathered}
 C_{\phi,H}
 =\bigvee_{\substack{uv\in E\\\phi(u)=\phi(v)}}x_{uv}
 \vee
 \bigvee_{\substack{uv\in E\\\phi(u)\ne\phi(v)\\
 \{\phi(u),\phi(v)\}\notin E(H)}}x_{uv},\\
 C=\bigwedge_{\substack{\phi\in\Phi\\H\in\mathcal H}}C_{\phi,H}.
 \end{gathered}
\]
Each clause asks for a certificate: an edge in $X$ whose endpoints
either collide under $\phi$ or map to a nonedge of $H$.

For completeness, let $X$ consist of the $m$ edges of a $k$-clique in $G$.
Fix any $\phi\in\Phi$ and $H\in\mathcal H$.
If two clique vertices collide under $\phi$, the edge between them
satisfies the first disjunction. Otherwise, their images form a
$k$-clique. Since $H$ has no $k$-clique, some edge of $X$ maps to a
nonedge of $H$, satisfying the second disjunction.
Thus every clause is satisfied, and $\OPT(C)\le m$.

For soundness, suppose $G$ has no $k$-clique and take any edge set
$X\subseteq E$ with $|X|\le h$.
Let $A$ be the set of endpoints of edges in $X$, so $|A|\le2h$.
Choose $\phi\in\Phi$ injective on $A$, and let $H$ be the graph on
$[2h]$ with edge set
\[
 E(H)=\{\{\phi(u),\phi(v)\}:uv\in X\}.
\]
This is a simple graph with at most $h$ edges. By injectivity, any
$k$-clique in $H$ would lift to one in $G$, so $H\in\mathcal H$.
No edge in $X$ has colliding endpoints, and every edge in $X$ maps
into $E(H)$. Both disjunctions in $C_{\phi,H}$ are therefore false.
Thus $\OPT(C)>h$, giving an $(m,h)$-gap \SetCover{} instance.

There are at most $(h+1)\binom{2h}{2}^{h}=2^{O(h\log h)}$
graphs to consider. Each can be tested for a $k$-clique in $h^{O(k)}$
time by enumerating $k$-subsets of $[2h]$.
Together with the fact that the perfect hash family can be constructed
in time $2^{O(h)}|V|\log(2|V|)$, this gives total construction time
$(|V|+n)^{O(1)}2^{O(h\log h)}$.
Taking $h=\lfloor\log n/\log\log n\rfloor$ makes this polynomial.
This proves \Cref{thm:main} and rules out $f(k)n^{o(\sqrt k)}$-time
algorithms under \hypoeth{}. To avoid the quadratic parameter blow-up,
we instead reduce from sparse parameterized 2-\CSP{}, as described in
\Cref{sec:reduction}.

\section{Preliminaries}\label{sec:prelim}

\begin{definition}[\SetCover{}]
A \SetCover{} instance is a bipartite graph $I=(\mathcal S,U,E)$ with a family
$\mathcal S$ of candidate sets, a universe $U$, and incidences
$E\subseteq\mathcal S\times U$.
A subfamily $X\subseteq\mathcal S$ is a \emph{cover} if every element of $U$
has a neighbor in $X$. Define
\[
 \OPT(I)=\min\{\,|X|:X\subseteq\mathcal S\text{ covers }U\,\},
\]
with $\OPT(I)=+\infty$ if no cover exists. 
The $k$-\SetCover{} problem asks whether $U$ can be covered by at most
$k$ sets from $\mathcal S$.
For $h\ge k$, the $(k,h)$-gap \SetCover{} problem asks us to distinguish
between the case $\OPT(I)\le k$ and the case $\OPT(I)>h$.
\end{definition}

\begin{definition}[Minimum monotone satisfying assignment at depth two]
An $\MMSA_2$ instance is a monotone CNF $C$ on $n$ Boolean variables:
an AND of clauses, each of which is an OR of unnegated Boolean variables. For a set $X$ of
Boolean variables, write $C(X)$ for the value of $C$ when precisely the Boolean variables in
$X$ are set to $1$. Define
\[
 \OPT(C)=\min\{\,|X|:C(X)=1\,\},
\]
with $\OPT(C)=+\infty$ if no satisfying assignment exists.
The $k$-$\MMSA_2$ problem asks whether $C$ has a satisfying assignment
with at most $k$ Boolean variables set to $1$.
\end{definition}

\SetCover{} is equivalent to $\MMSA_2$, as pointed out in~\cite{gm97}:
each Boolean variable corresponds to a candidate set, and each clause
corresponds to an element of the universe. This correspondence preserves
the optimum. Thus, in the rest of this paper, we use CNF notation to describe
instances of $k$-\SetCover{}.

\Needspace{12\baselineskip}

\begin{definition}[$k$-\Clique{}]\label{def:clique}
An instance of $k$-\Clique{} is a simple undirected graph $G=(V,E)$.
The goal is to decide whether there exist $k$ distinct vertices
$v_1,\ldots,v_k\in V$ such that $v_iv_j\in E$ for every $i<j$.
\end{definition}

\begin{definition}[Parameterized 2-\CSP{} {\cite{glrsw25}}]\label{def:csp}
An instance of parameterized 2-\CSP{} consists of an undirected constraint graph $G=(V,E)$,
a finite alphabet $\Sigma$, and an allowed relation $R_i\subseteq\Sigma^2$
for each constraint $i\in[m]$, where $m=|E|$.
Fix an ordering $(u_i,v_i)$ of the endpoints of each constraint $i$.
An assignment $a:V\to\Sigma$ satisfies the instance if
\[
 (a(u_i),a(v_i))\in R_i\qquad\text{for every }i\in[m].
\]
We write $k=|V|$ for the number of variables and $n=|\Sigma|$ for the alphabet size. We further assume that the constraint graph is simple and connected.
\end{definition}

\begin{lemma}\label{lem:csp-hardness}
Parameterized 2-\CSP{} is \wone{}-hard. Furthermore, it has no
$f(k)n^{o(k/\log k)}$-time algorithm for any computable function
$f$ under \hypoeth{}, even when the constraint graph is 3-regular.
\end{lemma}
\begin{proof}
The \wone{}-hardness follows from a reduction from
$k$-\Clique{}.
Given a $k$-\Clique{} instance $G$ with $k\ge2$, introduce $k$ variables, each choosing a vertex from a copy of $V(G)$. Let the alphabet be $\Sigma=V(G)$ and the constraint graph be $K_k$.
For each $i<j$, let
\[
 R_{ij}=\{(u,v)\in\Sigma^2:uv\in E(G)\}.
\]
It's easy to see that a satisfying assignment exists if and only if $G$ contains
a $k$-clique.

The \hypoeth{} lower bound follows from~\cite[Theorem~1.2]{kmps}, whose construction uses 3-regular expanders, which are connected, as constraint graphs.
\end{proof}

\begin{lemma}[Perfect hash families {\cite{ayz}}]\label{lem:hash}
For any nonempty finite set $V$ and $\Sigma=[h]$, there is a family
$\Phi(V,\Sigma)$ of maps $\phi:V\to\Sigma$ such that
\[
 \forall A\subseteq V,\ |A|\le h,
 \qquad \exists\phi\in\Phi(V,\Sigma)\text{ injective on }A.
\]
This family has size $2^{O(h)}\log(2|V|)$ and can be constructed deterministically
in $2^{O(h)}|V|\log(2|V|)$ time.
\end{lemma}

\section{A Reduction from Parameterized 2-CSP to $k$-\SetCover{}}\label{sec:reduction}

Let $\Gamma$ be an instance of parameterized 2-\CSP{} with constraint
graph $G=(V,E)$, alphabet $\Sigma$, and relations $R_1,\ldots,R_m$. 

\paragraph{The construction.}
Introduce one Boolean variable $x_{i,a,b}$ for each $(a,b)\in R_i$, giving a total of
\[
 N=\sum_{i=1}^m|R_i|
\]
Boolean variables. Setting $x_{i,a,b}=1$ means selecting the label pair $(a,b)$ for the $i$-th constraint.

Fix an integer $h\ge m$.
Take the perfect hash family $\Phi=\Phi(\Sigma,[2h])$ as constructed in \Cref{lem:hash}.
Define a \textit{template} to be a family of relations $H=(H_i)_{i\in[m]}$, where
$H_i\subseteq[2h]^2$ and $\sum_{i=1}^m|H_i|\le h$.
It specifies a 2-\CSP{} on the same constraint graph $G$, with
alphabet $[2h]$ and relation $H_i$ on $(u_i,v_i)$.
The template is \emph{bad} if this 2-\CSP{} is unsatisfiable, i.e., if there is no assignment $z:V(G)\to[2h]$ such that
\[
 (z(u_i),z(v_i))\in H_i\qquad\text{for every }i\in[m].
\]
For each $\phi\in\Phi$ and each bad template $H$, form the clause
\begin{equation}\label{eq:unified-circuit}
 C_{\phi,H}
 =\bigvee_{\substack{i\in[m]\\(a,b)\in R_i\\
 (\phi(a),\phi(b))\notin H_i}}x_{i,a,b},
\end{equation}
which asks for a selected label pair $(a,b)$ such that
$(\phi(a),\phi(b))\notin H_i$.

Let the final formula be
$$
 C=\bigwedge_{\substack{\phi\in\Phi\\H\text{ bad}}}C_{\phi,H}.
$$
In words, for every hash function $\phi\in\Phi$ and every bad template
$H$, we ask for a certificate that the hashed label pairs are \textit{not} contained in $H$.

\paragraph{Completeness.}
Let $\alpha:V(G)\to\Sigma$ be a satisfying assignment of $\Gamma$, and let 
\[
 X_\alpha=\{x_{i,\alpha(u_i),\alpha(v_i)}:i\in[m]\}.
\]
Set the $m$ Boolean variables in $X_\alpha$ to $1$ and all others to $0$.
Fix any $\phi\in\Phi$ and any bad template $H$.
Since $H$ is unsatisfiable, the assignment $z=\phi\circ\alpha$
violates some relation $H_i$. Hence
$(\phi(\alpha(u_i)),\phi(\alpha(v_i)))\notin H_i$, so the variable
$x_{i,\alpha(u_i),\alpha(v_i)}$, which is set to $1$, satisfies
$C_{\phi,H}$. This holds for every $\phi,H$, so $C(X_\alpha)=1$
and $\OPT(C)\le m$.

\paragraph{Soundness.}
Suppose $\Gamma$ is not satisfiable. Let $X$ be a set of size at most $h$ Boolean variables and let $$A=\{a,b:x_{i,a,b}\in X\} \subseteq \Sigma$$
be the collection of selected labels in $X$, we have $|A|\le2h$. Choose
$\phi\in\Phi$ injective on $A$. Define the template $H=(H_i)_{i\in[m]}$ where
\[
 H_i=\{(\phi(a),\phi(b)):x_{i,a,b}\in X\}
 \qquad\text{for each }i\in[m].
\]
Then $\sum_i|H_i|\le|X|\le h$.
If $H$ were satisfiable, let $z:V \to [2h]$ be a satisfying assignment. Since every vertex is incident to a constraint and $z$ satisfies $H$, we have $z(v)\in\phi(A)$ for every $v\in V$. Define $\alpha:V \to \Sigma$ by setting $\alpha(v)=(\phi|_A)^{-1}(z(v))$ for every $v \in V$.

Injectivity guarantees that $\alpha$ is well-defined, and that for every $i \in [m]$, $x_{i,\alpha(u_i),\alpha(v_i)} \in X$. Thus $\alpha$ satisfies $\Gamma$, a contradiction. Therefore, $H$ is bad and $C_{\phi,H}(X)=0$ by our definition of $H$. This holds for every $X$ of size at most $h$, so $\OPT(C)>h$.

\paragraph{Construction time.}
The number of Boolean variables is at most  $$N=\sum_{i=1}^m|R_i|\le m|\Sigma|^2.$$
Let $M=m(2h)^2$,
the number of templates is at most
\[
 \sum_{j=0}^h\binom Mj
 \le(h+1)M^h=2^{O(h\log h)}.
\]
To check whether a template is bad, we enumerate
one pair from each $H_i$ and check whether these pairs agree on shared variables. Since $\prod_i|H_i|\le2^{\sum_i|H_i|}\le2^h$,
this takes $2^{O(h)}$ time per template.

By \Cref{lem:hash}, the perfect hash family $\Phi$ has size
$2^{O(h)}\log(2|\Sigma|)$ and can be constructed in
$2^{O(h)}|\Sigma|\log(2|\Sigma|)$ time. Therefore, the entire circuit has
construction time and size at most
\begin{equation}\label{eq:unified-cost}
 N^{O(1)}2^{O(h\log h)}.
\end{equation}
For $N\ge16$, set
\begin{equation}\label{eq:nearlog-h}
 h=\left\lfloor\frac{\log N}{\log\log N}\right\rfloor.
\end{equation}
Since $h\log h\le\log N$, the construction is in polynomial-time.

Plugging the \wone{}-hardness and the \hypoeth{} lower bound for
parameterized 2-\CSP{} from \Cref{lem:csp-hardness} into the reduction,
we obtain the desired hardness of $k$-\SetCover{}.

\begin{proof}[Proof of \Cref{thm:main}]
Apply the reduction to the instance obtained from $k$-\Clique{} as in the proof of \Cref{lem:csp-hardness}.
Thus the $(m,h)$-gap \SetCover{} problem is \wone{}-hard, where $m=\binom{k}{2}$.
\end{proof}

\begin{proof}[Proof of \Cref{thm:main-eth}]
Apply the reduction to the 3-regular instances in \Cref{lem:csp-hardness}. The target parameter $m=3k/2$ due to the 3-regularity and the number of Boolean variables is $N \le n^{O(1)}$. Thus, an $f(m)N^{o(m/\log m)}$-time algorithm for $(m,h)$-gap \SetCover{}
would solve the source parameterized 2-CSP instances in $g(k)n^{o(k/\log k)}$ time
for some computable function $g$, contradicting \Cref{lem:csp-hardness}.
\end{proof}

\Needspace{8\baselineskip}
\begingroup\small\raggedright

\endgroup
\end{document}